\documentclass[runningheads]{llncs}
\usepackage[T1]{fontenc}
\usepackage{graphicx,verbatim}
\usepackage{multirow}
\usepackage{booktabs}
\usepackage{amsfonts} 
\usepackage{gensymb}
\usepackage{float}
\usepackage{capt-of}
\usepackage{centernot}
\usepackage{mathtools}
\usepackage{stmaryrd}
\usepackage{pifont}
\usepackage{hyperref}
\usepackage{bm}
\usepackage{algorithm}
\usepackage{algpseudocode}
\usepackage{amssymb} 
\usepackage{paralist}
\usepackage{colortbl}
\definecolor{lightgray}{gray}{0.9}    

\usepackage{subcaption}
\begin{document}
\title{Efficient Flow Matching for Sparse-View CT Reconstruction}
%
\author{Jiayang Shi\inst{1,2} \and
Lincen Yang\inst{2} \and
Zhong Li\inst{2,3} \and
Tristan van Leeuwen\inst{1,4} \and
Dani\"{e}l M. Pelt\inst{2} \and
K. Joost Batenburg\inst{2}
}
\authorrunning{J. Shi et al.}
%
\institute{Centrum Wiskunde en Informatica \and
LIACS, Leiden University \and
Great Bay University \and
Mathematical Institute, Utrecht University\\
\email{l.yang@liacs.leidenuniv.nl}}


  
\maketitle              
\begin{abstract}
Generative models, particularly Diffusion Models (DM), have shown strong potential for Computed Tomography (CT) reconstruction serving as expressive priors for solving ill-posed inverse problems. However, diffusion-based reconstruction relies on Stochastic Differential Equations (SDEs) for forward diffusion and reverse denoising, where such stochasticity can interfere with repeated data consistency corrections in CT reconstruction. Since CT reconstruction is often time-critical in clinical and interventional scenarios, improving reconstruction efficiency is essential. In contrast, Flow Matching (FM) models sampling as a deterministic Ordinary Differential Equation (ODE), yielding smooth trajectories without stochastic noise injection. This deterministic formulation is naturally compatible with repeated data consistency operations. Furthermore, we observe that FM-predicted velocity fields exhibit strong correlations across adjacent steps. Motivated by this, we propose an FM-based CT reconstruction framework (FMCT) and an efficient variant (EFMCT) that reuses previously predicted velocity fields over consecutive steps to substantially reduce the number of Neural network Function Evaluations (NFEs), thereby improving inference efficiency. We provide theoretical analysis showing that the error introduced by velocity reuse is bounded when combined with data consistency operations. Extensive experiments demonstrate that FMCT/EFMCT achieve competitive reconstruction quality while significantly improving computational efficiency compared with diffusion-based methods. The codebase is open-sourced at \href{https://github.com/EFMCT/EFMCT}{https://github.com/EFMCT/EFMCT}.


\end{abstract}
\section{Introduction}
Generative models have recently gained popularity for Computed Tomography (CT) reconstruction due to their strong capability to learn expressive image priors \cite{ying2019x2ct,liu2023solving,liu2023dolce,chen2025cross}. Among them, Diffusion Models (DM) have achieved state-of-the-art (SOTA) performance in sparse-view CT reconstruction \cite{liu2023dolce,chen2025cross,shi2026dmct} by learning the distribution of full-view reconstructions and leveraging it as a powerful prior for solving ill-posed inverse problems. Despite their strong reconstruction quality, DM methods are inherently stochastic because they model the forward diffusion and reverse denoising processes as Stochastic Differential Equations (SDEs) \cite{song2021scorebased}. When combined with inverse problem solving, the repeated data consistency correction steps required for CT reconstruction can interfere with the stochastic evolution of the SDE \cite{zhang2025decoupling,chung2023diffusion}, leading to a \textit{push-and-pull effect} between prior-driven denoising and physics-based correction. This interaction often results in unstable or inconsistent reconstruction behavior \cite{shi2026dmct,chung2023diffusion}, and necessitates a large number of iterations and Neural network Function Evaluations (NFEs), typically on the order of thousands. Although deterministic schedulers such as DDIM \cite{song2021denoising} have been proposed to reduce the number of reverse steps, achieving high-quality reconstructions with substantially fewer iterations remains challenging due to the inherently stochastic formulation of diffusion models.

In clinical practice, CT reconstruction is frequently time-critical. In emergency and interventional settings, rapid image availability directly impacts diagnosis and treatment decisions, making time to first image a key performance metric \cite{fanucci2007whole,willemink2013iterative}. In addition, modern medical CT systems increasingly target higher spatial resolutions, which lead to larger reconstruction volumes and further exacerbate computational demands \cite{hussain2022modern}. While conventional Filtered Backprojection (FBP) provides fast reconstructions, its image quality degrades substantially under sparse-view or low-dose acquisition, where advanced prior-based methods become necessary. As a result, the excessive computational cost of diffusion-based reconstruction methods poses a significant challenge for their practical deployment in real-world clinical workflows \cite{shi2026dmct}.

Flow Matching (FM) offers a fundamentally different modeling perspective that is particularly well suited to this setting. Instead of stochastic sampling, FM models deterministic probability transport through an Ordinary Differential Equation (ODE) \cite{lipman2023flow}. Consequently, FM-based sampling follows a constant trajectory without stochastic noise injection. This deterministic formulation enables more coherent interaction with repeated data consistency corrections in inverse problems \cite{zhang2024flow,pourya2026flower}, thereby mitigating the \textit{push-and-pull effect} observed in diffusion-based methods and allowing for more efficient sampling. Furthermore, the velocity fields learned by FM are observed to evolve smoothly along the ODE trajectory, with strong correlations between adjacent time steps, as demonstrated in \cite{bajpai2026flowcast} and empirically confirmed in our experiments (see Fig.~\ref{fig:cosine_nfe}). This observation suggests that repeatedly re-evaluating the neural network at every integration step may be redundant.

Motivated by these properties, we propose an efficient FM-based CT reconstruction framework that reuses previously predicted velocity fields to substantially reduce the number of NFEs. The reuse of velocity fields introduces only a controlled integration error, which we show is of the same order as the Euler discretization error and remains bounded through the explicit data consistency corrections. As a result, the proposed strategy achieves significant efficiency gains while maintaining competitive reconstruction quality. Our main contributions are summarized as follows:
\begin{inparaenum}[1)]
\item We propose, to the best of our knowledge, the first flow-matching-based framework for CT reconstruction.
\item We introduce a velocity field reuse strategy that significantly improves the efficiency of FM-based CT reconstruction by reducing the required NFEs.
\item We provably show that the single-step reuse error is of the same order as Euler discretization, and does not alter the overall convergence behavior under bounded consecutive reuse steps.
\item We conduct extensive experiments to demonstrate the reconstruction performance and efficiency of the proposed FM-based CT reconstruction method.
\end{inparaenum}

\begin{figure}[t]
\centering
\includegraphics[width=\textwidth]{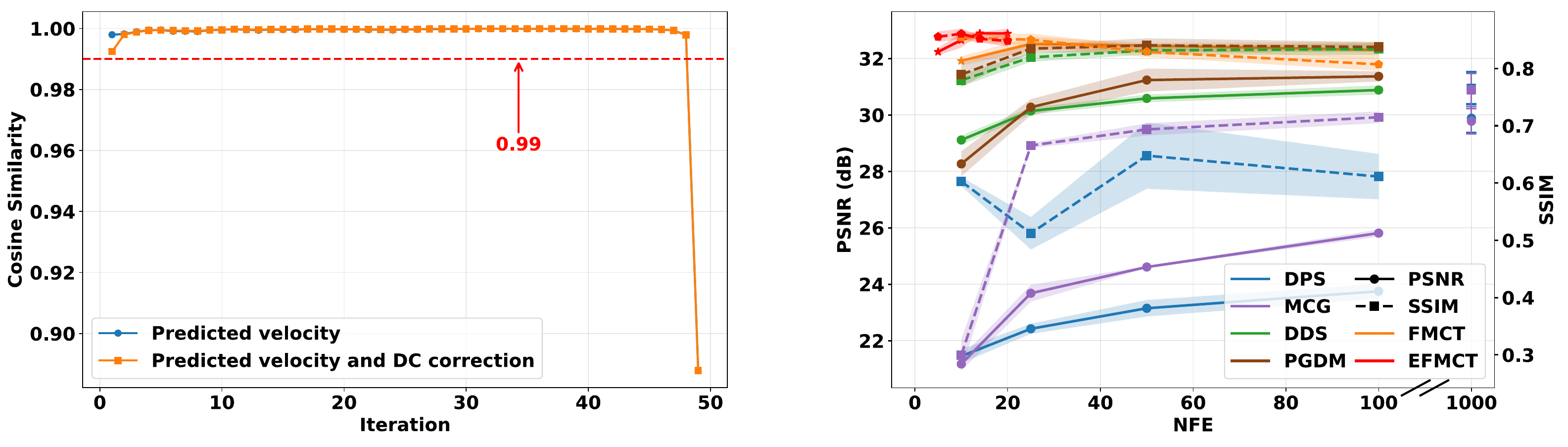}
\caption{\textbf{Left}: Cosine similarity of predicted velocity fields at consecutive iterations in FMCT, with and without data consistency correction, showing strong correlation across adjacent steps. \textbf{Right}: Reconstruction performance vs. NFE for DM methods and FMCT/EFMCT, demonstrating that FMCT/EFMCT achieve competitive quality at substantially lower NFEs. DM methods here use the deterministic DDIM sampler for fairness; DPS/MCG results with the \textit{original 1000-step DDPM sampler} are shown as \textit{isolated points}. Results are averaged over the same randomly selected 21 reconstructions from the AAPM dataset across all iteration settings; shaded regions indicate standard deviation.}
\label{fig:cosine_nfe}
\end{figure}

\section{Method}
\textbf{Problem Formulation.} CT reconstruction aims to recover an unknown object $\bm{x} \in \mathbb{R}^m$ from a set of projection measurements $\bm{y} \in \mathbb{R}^n$. The measurement process can be mathematically modeled as a linear system 
\begin{equation}
    \bm{y} = \bm{A}\bm{x},
\end{equation}
where $\bm{A} \in \mathbb{R}^{n \times m}$ is the system matrix determined by the acquisition geometry. In sparse-view CT, the number of measurements is insufficient ($n < m$), making the inverse problem underdetermined and ill-posed. In addition, measurement noise further degrades reconstruction quality. In this work, we focus on the sparse-view aspect of CT reconstruction.

\noindent \textbf{Flow Maching.} Given a source distribution $p(\bm{x}_1)$ and a target distribution $p(\bm{x}_0)$, FM defines a time-dependent velocity field $\bm{v}_t: \mathbb{R}^d\rightarrow\mathbb{R}^d$
that transports samples from $p(\bm{x}_1)$ to $p(\bm{x}_0)$ via the ODE \cite{lipman2023flow}
\begin{equation}
    \frac{d\bm{x}_t}{dt}
    =
    \bm{v}_t(\bm{x}_t),
    \quad t\in[0,1].
\end{equation}
The idea of FM is to train to approximate $\bm{v}_t$ by a neural network $\bm{v}_t^{\bm{\theta}}$, which allows sampling from the target distribution by numerically integrating the ODE. The optimization objective of FM then becomes
\begin{equation}
\label{eq:fm_object}
    \mathcal{L}_{FM}= \mathbb{E}_{t\sim\mathcal{U}(0,1), \: \bm{x}_t \sim p(\bm{x}_t)}\| \bm{v}_t(\bm{x}_t)-\bm{v}_t^{\bm{\theta}}(\bm{x}_t)\|^2.
\end{equation}
However, the marginal velocity field $\bm{v}_t(\bm{x}_t)$ is generally intractable, as it requires integration over all possible $\bm{x}_0$. The common way is to replace the marginal velocity with a conditional velocity field $\bm{v}_t(\bm{x}_t|\bm{x}_0)$, which is proven in  \cite{lipman2023flow} to yield gradients equivalent to those of Equation \ref{eq:fm_object}.

In this work, we focus on one specific conditional linear flow model, also known as Rectified Flow \cite{liu2023flow}. It models the transportation along straight-line trajectories through interpolation between $\bm{x}_0$ and $\bm{x}_1$
\begin{equation}
    \bm{x}_t = (1-t)\bm{x}_0+t\bm{x}_1,   \quad \bm{v}_t(\bm{x}_t|\bm{x}_0) = \bm{x}_1-\bm{x}_0.
\end{equation} This leads to the conditional flow matching loss
\begin{equation}
\label{eq:cfm_object}
    \mathcal{L}_{CFM}= \mathbb{E}_{t\sim\mathcal{U}(0,1), \: \bm{x}_0\sim p(\bm{x}_0), \: \bm{x}_1\sim p(\bm{x}_1)}\| \bm{v}^{\bm{\theta}}_t((1-t)\bm{x}_0+t\bm{x}_1)- (\bm{x}_1-\bm{x}_0)\|^2.
\end{equation}
In practice, we sample $\bm{x}_1$ from a standard Gaussian distribution and use full-view CT images as $\bm{x}_0$. Once trained, the model enables deterministic sampling by integrating the learned velocity field backward from $t=1$ to $t=0$.

\begin{figure}[t]
\centering
\includegraphics[width=1\textwidth]{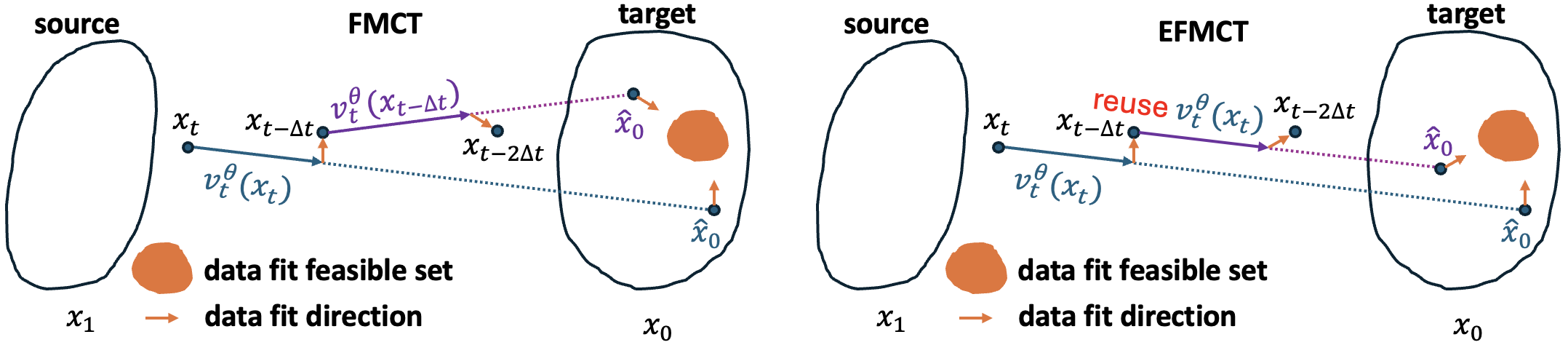}
\caption{Overview of our proposed FMCT(left) and EFMCT(right) method.}
\label{fig:method}
\end{figure}

\noindent \textbf{Flow Matching for CT Reconstruction.} Given a pretrained FM model, we incorporate physics-based data consistency into the sampling process to enable CT reconstruction. Each integration step consists of three components.

First, a flow transport step is performed using an explicit Euler update 
\begin{equation}
\label{eq:efmct_first}
\bm{x}_{t-\Delta t}'=\bm{x}_t-\Delta t\bm{v}_t^{\bm{\theta}}(\bm{x}_t). 
\end{equation}

Second, we estimate the terminal sample at $t=0$ via linear extrapolation, and apply the data consistency operation $\mathcal{DC}(\cdot)$ to obtain the correction direction $\bm{v}_{\mathcal{DC}}$ from measurement aware refinement 
\begin{equation}
\label{eq:efmct_second}
\hat{\bm{x}}_0=\bm{x}_t-t\bm{v}_t^{\theta}(\bm{x}_t), \quad \bm{v}_{\mathcal{DC}}=\mathcal{DC}(\hat{\bm{x}}_0,\bm{y},\bm{A}). 
\end{equation}
In this work, we implement $\mathcal{DC}(\cdot)$ using conjugate gradient \cite{fletcher1964function,chung2024decomposed}, though the framework is compatible with other data consistency schemes. 

Finally, the data-consistent update is applied as
\begin{equation}
\label{eq:efmct_third}
\bm{x}_{t-\Delta t}=\bm{x}_{t-\Delta t}'+\bm{v}_{\mathcal{DC}}.
\end{equation}
These steps are repeated iteratively until $t=0$ to obtain the final reconstruction. An overview of the proposed framework is illustrated in Fig.~\ref{fig:method}.

\noindent \textbf{Velocity Reuse for Efficient Sampling.} Empirically, we observe that the predicted velocities $\bm{v}_t^{\bm{\theta}}$ are highly correlated in consecutive steps (Fig. \ref{fig:cosine_nfe}), resulting in near-straight trajectories toward the reconstructed image. This behavior is consistent with the smooth evolution of velocity fields along the underlying ODE in flow matching. These observations suggest that repeatedly re-evaluating the neural network at every integration step may be redundant. Moreover, when combined with explicit data consistency operations, we argue that the small deviations introduced by reusing a previously predicted velocity can be effectively compensated by physics-based corrections, as conceptually illustrated in Fig.\ref{fig:method}.

Specifically, after computing $\bm{v}_t = \bm{v}_t^{\bm{\theta}}(\bm{x}_t)$ at time $t$, we reuse this velocity for up to $M$ consecutive steps:
\begin{equation}
\label{eq:velocity_reuse}
\bm{x}_{t-(j+1)\Delta t}'=\bm{x}_{t-j\Delta t}-\Delta t\bm{v}_t, \quad j \in[1,M].
\end{equation}
Since $\bm{v}_t$ is reused without additional network evaluations, this strategy directly reduces the number of NFEs. To ensure stability and data fidelity, we perform an adaptive refinement check at each reuse step
\begin{equation}
\|\bm{A}\bm{x}_{t-(j+1)\Delta t}'-\bm{y}\|^2 \leq \eta\|\bm{A}\bm{x}_{t-j\Delta t}'-\bm{y}\|^2,    
\end{equation}
where $\eta>1$ is a relaxation factor. If the condition is violated, velocity reuse is terminated and the velocity is recomputed using the neural network. This adaptive mechanism balances efficiency and reconstruction accuracy. $\eta$ is emprically set at $1.05$. The complete algorithm in summarized in the appendix.

\begin{proposition} \label{prop:reuse}
Assume that $\bm{v}$ is locally Lipschitz in both $\bm{x}$ and t, i.e.,  $\|\bm{v}_t(\bm{x})-\bm{v}_t(\bm{x}')\|  \le L_x \| \bm{x}-\bm{x}' \|$ and $\|\bm{v}_t(\bm{x})-\bm{v}_{t'}(\bm{x})\|  \le L_t | t-t' |$ for all relevant $\bm{x},\bm{x}',t,t'$.
Then, for a single reuse step, the local deviation between the reuse update $\tilde{\bm{x}}_{k+1}$ and the standard Euler update $\bm{x}_{k+1}$ satisfies $\|\tilde{\bm{x}}_{k+1}-\bm{x}_{k+1}\| = O(\Delta t^2)$, i.e., velocity reuse introduces a local error of the same order as the Euler discretization itself. Moreover, if the same velocity is reused for at most $M$ consecutive steps (independent of $\Delta t$),
and the data-consistency correction is non-expansive with respect to its image argument, then the deviation remains controlled. The accumulated deviation after one reuse block is $O(M^2\Delta t)$.
\end{proposition}
\noindent \textbf{Implication.}
The proposition implies that velocity reuse introduces a local error of \textit{the same order as the Euler discretization itself}. When the reuse length is bounded, the accumulated error remains controlled over consecutive reuse steps. \textbf{Proofs are in Appendix.}

\section{Results}
\textbf{Comparison Methods.} We compare the proposed method with several representative diffusion-based CT reconstruction approaches, including Diffusion Posterior Sampling (DPS) \cite{chung2023diffusion}, Manifold Constrained Gradient (MCG) \cite{chung2022improving}, Pseudoinverse-Guided Diffusion Models (PGDM) \cite{song2023pseudoinverseguided}, and Decomposed Diffusion Sampler (DDS) \cite{chung2024decomposed}. These methods represent a diverse set of strategies for incorporating measurement information into diffusion-based reconstruction and serve as strong SOTA baselines. In addition to learning-based methods, we include two classical reconstruction approaches, FBP and Model-Based Iterative Reconstruction (MBIR) method, Alternating Direction Method of Multipliers with Split-Bregman Total Variation (ADMM-TV) \cite{boyd2011distributed,goldstein2009split}. 

\begin{table}[t]
\setlength{\tabcolsep}{3pt}
\centering
\caption{Quantitative comparison of reconstruction performance across different methods. Mean and standard deviation are reported for PSNR, SSIM, and data fidelity. The best average result for each metric is shown in \textbf{bold}, and the second best is \underline{underlined}. NFE$^*$ counts only neural network forward evaluations. $\downarrow$ indicates efficiency improvement in NFE and computation time (on RTX4090) of FMCT/EFMCT compared with the most efficient diffusion-based method.}
\resizebox{\textwidth}{!}{
\begin{tabular}{c c ccc ccc cc}
\toprule[0.1pt]
\multirow{2}{*}{Dataset} 
& \multirow{2}{*}{Method}
& \multicolumn{3}{c}{40 views} 
& \multicolumn{3}{c}{20 views} 
& \multirow{2}{*}{NFE$^*$} 
& \multirow{2}{*}{Time/s} \\

\cmidrule(lr){3-5} \cmidrule(lr){6-8}
& 
& PSNR & SSIM & Data Fit 
& PSNR & SSIM & Data Fit 
&  &  \\

\midrule[0.05pt]

\multirow{8}{*}{\rotatebox[origin=c]{90}{\textbf{AAPM}}}
& FBP   
& $26.98 \pm 0.45$ & $0.691 \pm 0.027$ & $2783.64 \pm 233.30$
& $24.01 \pm 0.49$ & $0.547 \pm 0.029$ & $1798.77 \pm 171.00$
& - & 0.02 \\  

& ADMM-TV   
& $\textbf{31.94} \pm 0.42$ & $\textbf{0.836} \pm 0.027$ & $\underline{22.11} \pm 2.29$
& $29.45 \pm 0.46$ & $\textbf{0.798} \pm 0.029$ & $\underline{17.53} \pm 2.00$
& - & 9.85 \\  

& DPS   
& $30.05 \pm 0.69$ & $0.795 \pm 0.033$ & $261.35 \pm 20.58$
& $28.78 \pm 0.81$ & $0.764 \pm 0.032$ & $142.93 \pm 35.99$
& 1000 & 142.16 \\  

& MCG   
& $30.00 \pm 0.79$ & $0.791 \pm 0.032$ & $198.07 \pm 14.20$
& $27.57 \pm 0.42$ & $0.736 \pm 0.031$ & $180.83 \pm 11.87$
& 1000 & 143.44 \\

& PGDM  
& $30.26 \pm 0.61$ & $0.803 \pm 0.030$ & $208.98 \pm 13.72$
& $29.43 \pm 0.60$ & $0.775 \pm 0.031$ & $106.51 \pm 5.30$
& 100 & 29.79 \\

& DDS   
& $31.06 \pm 0.40$ & $\underline{0.834} \pm 0.030$ & $\textbf{19.09} \pm 1.14$
& $29.21\pm 0.36$ & $0.785 \pm 0.031$ & $\textbf{9.87} \pm 0.61$
& 50 & 3.83 \\

\cmidrule(lr){2-10}

& \cellcolor{lightgray} FMCT  
& \cellcolor{lightgray} $\underline{31.63} \pm 0.39$ & \cellcolor{lightgray} $0.810 \pm 0.031$ & \cellcolor{lightgray} $\underline{60.57} \pm 18.52$
& \cellcolor{lightgray} $\textbf{29.98} \pm 0.23$ & \cellcolor{lightgray} $0.769 \pm 0.030$ & \cellcolor{lightgray} $46.46 \pm 14.86$
& \cellcolor{lightgray} $\underline{25}(\downarrow50\%)$ & \cellcolor{lightgray}$\underline{1.92}(\downarrow50\%)$\\ 

& \cellcolor{lightgray} EFMCT 
& \cellcolor{lightgray} $31.52 \pm 0.28$ & \cellcolor{lightgray} $0.832 \pm 0.027$ & \cellcolor{lightgray} $182.80 \pm 104.39$
& \cellcolor{lightgray} $\underline{29.52} \pm 0.23$ & \cellcolor{lightgray} $\underline{0.795} \pm 0.027$ & \cellcolor{lightgray} $111.41 \pm 79.80$
& \cellcolor{lightgray} $\textbf{7}(\downarrow75\%)$ & \cellcolor{lightgray} $\textbf{0.83} (\downarrow78\%)$ \\  

\midrule[0.1pt]
\multirow{8}{*}{\rotatebox[origin=c]{90}{\textbf{Decathlon}}}
& FBP   
& $28.01 \pm 0.76$ & $0.763 \pm 0.037$ & $2772.79 \pm 434.36$
& $24.85 \pm 0.96$ & $0.618 \pm 0.039$ & $1731.97 \pm 305.77$
& - & 0.02 \\  

& ADMM-TV   
& $34.40 \pm 1.23$ & $0.900 \pm 0.021$ & $\underline{15.73} \pm 0.55$
& $30.26 \pm 1.03$ & $0.851 \pm 0.028$ & $\underline{14.76} \pm 0.50$
& - & 9.85 \\  

& DPS   
& $33.85 \pm 1.25$ & $0.870 \pm 0.018$ & $250.11 \pm 18.61$
& $31.69 \pm 1.44$ & $0.842 \pm 0.025$ & $121.25 \pm 18.57$
& 1000 & 142.16 \\  

& MCG   
& $30.03 \pm 1.20$ & $0.824 \pm 0.030$ & $301.17\pm 35.25$
& $28.67 \pm 1.26$ & $0.798 \pm 0.038$ & $198.22 \pm 22.68$
& 1000 & 143.44 \\

& PGDM  
& $34.78 \pm 1.23$ & $0.899 \pm 0.014$ & $147.10 \pm 15.20$
& $31.82 \pm 1.19$ & $0.860 \pm 0.018$ & $132.96 \pm 13.69$
& 100 & 29.79 \\

& DDS   
& $\textbf{37.66} \pm 1.20$ & $\underline{0.923} \pm 0.009$ & $\textbf{15.10} \pm 9.29$
& $33.00\pm 1.29$ & $\underline{0.873} \pm 0.018$ & $\textbf{5.01} \pm 0.36$
& 100 & 7.65 \\

\cmidrule(lr){2-10}

& \cellcolor{lightgray} FMCT  
& \cellcolor{lightgray} $\underline{37.65} \pm 1.53$ & \cellcolor{lightgray} $0.912 \pm 0.020$ & \cellcolor{lightgray} $334.44 \pm 74.21$
& \cellcolor{lightgray} $\textbf{33.41} \pm 1.57$ & \cellcolor{lightgray} $0.841 \pm 0.036$ & \cellcolor{lightgray} $294.97 \pm 58.87$
& \cellcolor{lightgray} $\underline{50}(\downarrow50\%)$ & \cellcolor{lightgray} $\underline{4.09} (\downarrow47\%)$ \\ 

& \cellcolor{lightgray} EFMCT 
& \cellcolor{lightgray} $37.56 \pm 1.37$ & \cellcolor{lightgray} $\textbf{0.940} \pm 0.014$ & \cellcolor{lightgray} $621.99 \pm 48.99$
& \cellcolor{lightgray} $\underline{33.21} \pm 1.29$ & \cellcolor{lightgray} $\textbf{0.888} \pm 0.024$ & \cellcolor{lightgray} $470.58 \pm 33.62$
& \cellcolor{lightgray} $\textbf{11}(\downarrow89\%)$ & \cellcolor{lightgray} $\textbf{1.72} (\downarrow78\%)$ \\  
\bottomrule[0.1pt]
\end{tabular}
}
\label{tab:result}
\end{table}

\begin{figure}[t]
\centering
\includegraphics[width=1\textwidth]{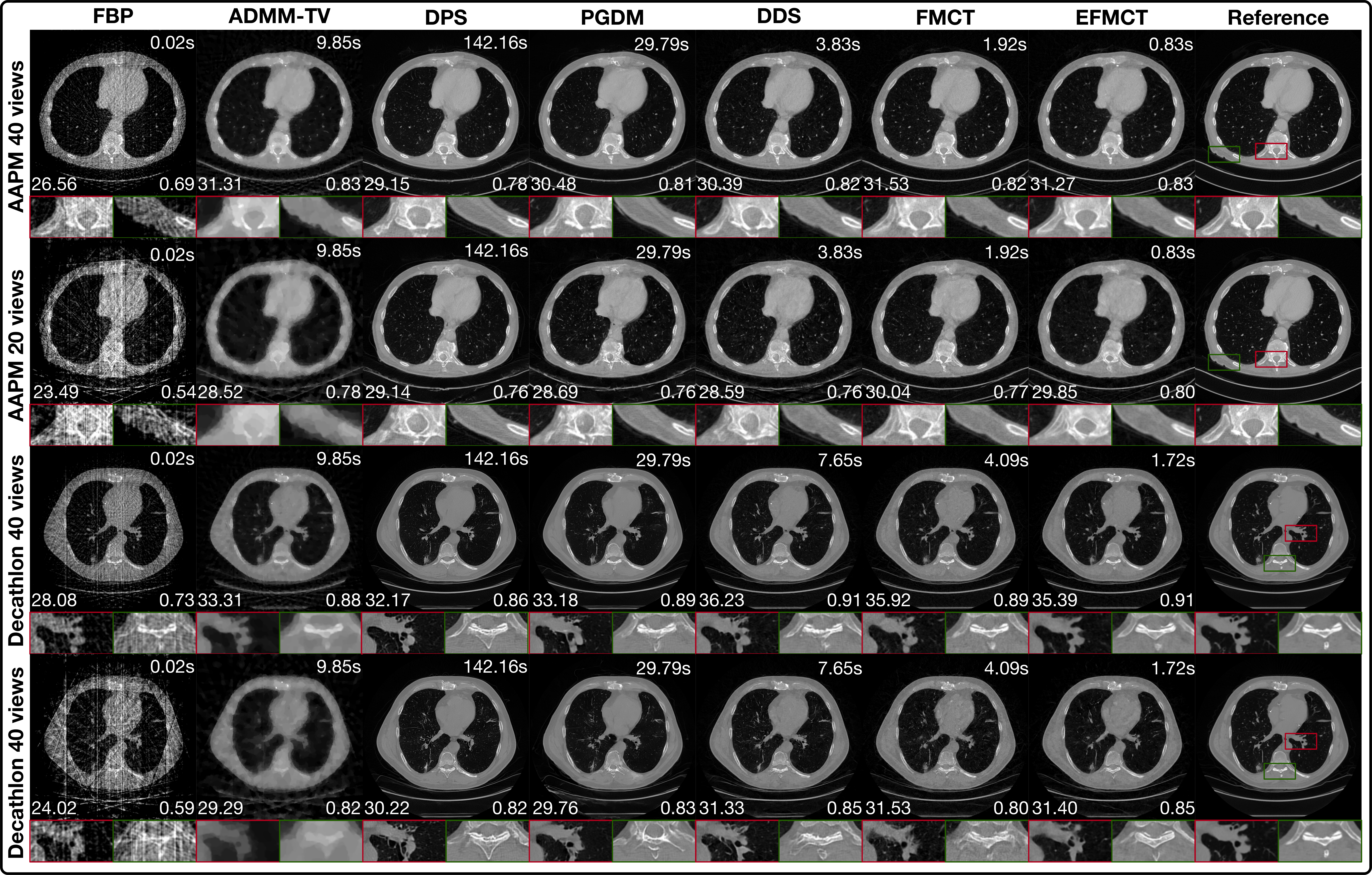}
\caption{Visual comparison of reconstructions across different methods, datasets, and views. PSNR, SSIM, and reconstruction time are shown in the lower-left, lower-right, and upper-right corners of each image, respectively. We visualize only the best three diffusion-based methods for better visibility. Zoom in for more details.}
\label{fig:result}
\end{figure}

\noindent \textbf{Datasets.} We conduct experiments on two publicly available CT datasets to evaluate the reconstruction performance of the proposed method. The first dataset is the AAPM 2016 Low Dose CT Grand Challenge dataset \cite{mccollough2017low}. From this dataset, CT images ($512\times512$ pixels) from nine patients are used to train the diffusion and flow matching models, while images from patient \textit{L506} are reserved for testing. The second dataset is the Medical Segmentation Decathlon CT dataset \cite{antonelli2022medical}. From the \textit{Task\_06 Lung} subset ($512\times512$ pixels), we randomly select ten patients for training and evaluate reconstruction performance on images from patient \textit{017}. To obtain sparse-view measurements, we simulate 20-view and 40-view projection data using the ASTRA Toolbox \cite{van2016fast}. All experiments are conducted using a parallel-beam geometry. This choice is made deliberately to provide a controlled and widely used experimental setting that minimizes confounding effects introduced by geometric complexity and allows us to isolate the impact of the reconstruction methodology itself. The proposed FMCT and EFMCT with velocity reuse strategy do not rely on geometry-specific assumptions and can be readily extended to other acquisition geometries by modifying the forward operator.

\noindent \textbf{Reconstruction Performance.} Table~\ref{tab:result} reports the quantitative reconstruction results. FMCT/EFMCT are often among the highest PSNR and SSIM across most datasets and view configurations, while requiring substantially fewer NFEs and lower computation time than diffusion-based baselines. ADMM-TV and DDS achieve the best and second-best data fidelity, reflecting their strong emphasis on measurement consistency. Compared with FMCT, EFMCT exhibits a slight reduction in PSNR and data fidelity but attains comparable or improved SSIM, accompanied by a significant efficiency improvement. This demonstrates that the proposed velocity reuse strategy effectively reduces computational cost while preserving reconstruction quality.

Figure~\ref{fig:result} provides a visual comparison. Although ADMM-TV achieves high PSNR/SSIM, its reconstructions are overly smooth and lack fine details. Diffusion-based methods and FMCT/EFMCT produce visually sharper and more detailed reconstructions, though they do not always exactly match the reference. This behavior is expected in sparse-view CT, where the problem is severely ill-posed and multiple reconstructions can satisfy the measurements. Generative methods balance data consistency with adherence to the learned prior, which may lead to small deviations in fine-scale details.

\begin{figure}[t]
\centering
\includegraphics[width=1\textwidth]{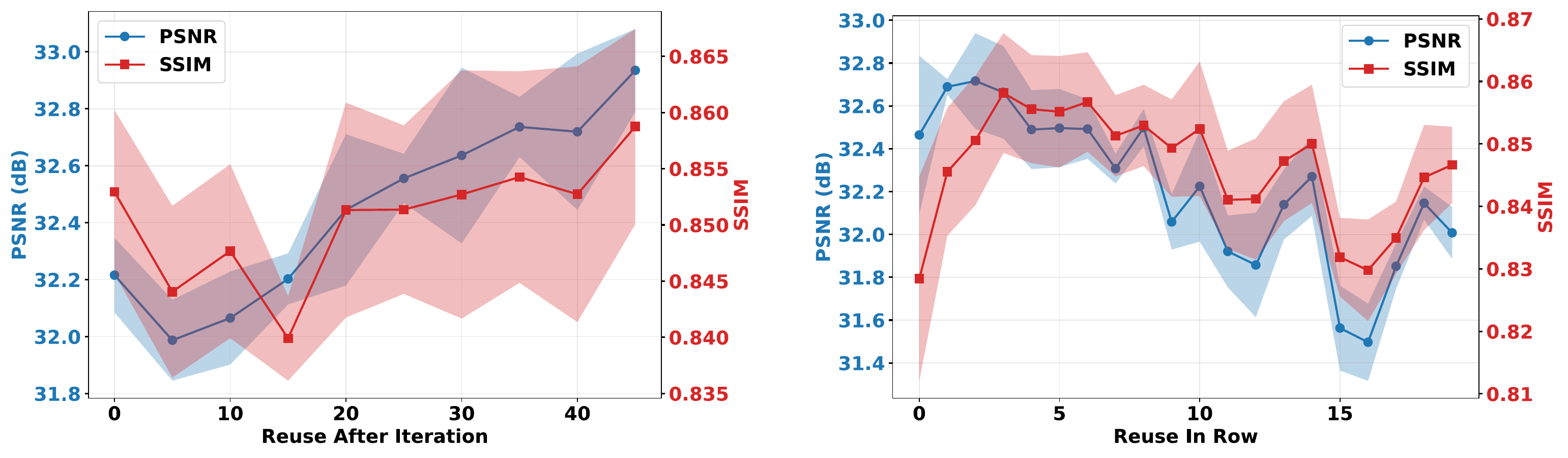}
\caption{Ablation of the velocity reuse strategy. Left: Reuse enabled after different iteration indices.
Right: Varying the maximum number of consecutive reuse steps. Shaded regions indicate standard deviation.}
\label{fig:ablation}
\end{figure}

\noindent \textbf{Impact of Velocity Reuse Strategy.} We analyze the effect of the proposed velocity reuse strategy through an ablation study shown in Fig.~\ref{fig:ablation}. The experiments are conducted on 21 randomly selected test CT images, with the total number of sampling iterations fixed to 50. In the first ablation, we vary the iteration at which velocity reuse is enabled while fixing the maximum number of consecutive reuse steps to 10. Enabling reuse too early leads to an initial drop in reconstruction performance, whereas delaying reuse to later iterations gradually improves PSNR and SSIM. This behavior suggests that early-stage sampling is more sensitive to approximation errors, while reuse becomes safer once the reconstruction has approached a more stable trajectory. Consequently, velocity reuse is better suited to mid- or late-stage sampling. In the second ablation, velocity reuse is enabled from the first iteration, and the maximum number of allowed consecutive reuse steps is varied. Increasing the reuse limit initially improves PSNR and SSIM, indicating that moderate reuse does not adversely affect reconstruction quality. However, allowing excessive consecutive reuse results in a general performance degradation, revealing a trade-off between reconstruction quality and computational efficiency. Aggressive reuse substantially reduces computational cost but may compromise reconstruction accuracy. For all EFMCT results, we adopt a reuse strategy that enables velocity reuse after the first iteration and allows up to 10 consecutive reuse steps in all experiments.


\section{Conclusion}
In this work, we proposed FMCT/EFMCT, representing the first application of flow matching to CT reconstruction. By exploiting the deterministic transport structure of flow matching, we introduced a velocity reuse strategy that substantially reduces NFEs. We provided theoretical analysis showing that the error introduced by velocity reuse is of the same order as Euler discretization and remains bounded when combined with explicit data consistency operations. Extensive experiments on multiple datasets demonstrate that FMCT/EFMCT achieve competitive reconstruction quality while significantly improving computational efficiency compared with diffusion-based methods. We believe that the proposed framework offers a practical and principled approach toward efficient generative CT reconstruction, with the potential to facilitate the deployment of advanced reconstruction methodologies in real-world clinical settings.\\

\noindent\textbf{Ackownledgements.} The authors are supported by European Union H2020-MSCA-ITN-2020 under grant agreement no. 956172 and Dutch Research Council under grant no. ENWSS.2018.003, NWA.1160.18.316.

\noindent\textbf{Disclosure of Interests.} The authors have no competing interests to declare that are relevant to the content of this article.
%
%
%
\bibliographystyle{splncs04}
\bibliography{mybibliography}

\newpage
\section{Appendix: Proof of Proposition~\ref{prop:reuse}}

\begin{proof}
Let $\bm{v}$ $\mathbb{R}^d\to\mathbb{R}^d$ be a velocity field. Define the explicit noise-to-image Euler update as
\begin{equation}
    \bm{x}_{k+1} = \bm{x}_k - \Delta t\, \bm{v}_{t_k}(\bm{x}_k),
\quad t_{k+1}=t_k-\Delta t.
\end{equation}
Define the velocity-reuse update for $M$ consecutive steps using a previously evaluated velocity computed at $(t_{k_0},\bm{x}_{k_0})$:
\begin{equation}
    \tilde{\bm{x}}_{k+1} = \bm{x}_k - \Delta t\, \bm{v}_{t_{k_0}}(\bm{x}_{k_0}),
\quad k=k_0,...,k_0+M-1.
\end{equation}

Assume Lipschitz continuity $L_x,L_t\ge0$ such that for all relevant $\bm{x},\bm{x}',t,t'$,
\begin{equation}
\|\bm{v}_t(\bm{x})-\bm{v}_t(\bm{x}')\|  \le L_x \| \bm{x}-\bm{x}' \|, \quad \|\bm{v}_t(\bm{x})-\bm{v}_{t'}(\bm{x})\|  \le L_t | t-t' |.
\end{equation}

\textbf{Single Step Error.} Assume reuse the velocity at the first step $k_0=k-1$ gives error $e_k$, 
\begin{equation}
\begin{aligned}
e_k
&= \|\bm{v}_{t_k}(\bm{x}_k)-\bm{v}_{t_{k-1}}(\bm{x}_{k-1})\| \\
&\le \|\bm{v}_{t_k}(\bm{x}_k)-\bm{v}_{t_{k-1}}(\bm{x}_{k})\|
   + \|\bm{v}_{t_{k-1}}(\bm{x}_{k})-\bm{v}_{t_{k-1}}(\bm{x}_{k-1})\| \\
&\le L_t |t_k - t_{k-1}| + L_x\|\bm{x}_{k}-\bm{x}_{k-1}\| \\
&\le L_t \Delta t+ L_x\Delta t \|\bm{v}_{t_{k-1}}(\bm{x}_{k-1})\|. 
\end{aligned}
\end{equation}
Then the error in $\tilde{\bm{x}}_{k+1}$ and $\bm{x}_{k+1}$ becomes
\begin{equation}
\begin{aligned}
\|\bm{x}_{k+1}-\tilde{\bm{x}}_{k+1}\|
& = \Delta t \|\bm{v}_{t_k}(\bm{x}_k)-\bm{v}_{t_{k-1}}(\bm{x}_{k-1})\| \\
& \le \Delta t^2  (L_t + L_x \|\bm{v}_{t_{k-1}}(\bm{x}_{k-1})\|). 
\end{aligned}
\end{equation}
It shows that the reusing the velocity gives single step error of $O(\Delta t^2)$, which matches the error of standard Euler $x(t-\Delta t)=x(t)-\Delta t\,\dot{x}+O(\Delta t^2)$.

\textbf{Accumulated Error.}
According to our proposed algorithm, we define standard single-step state update as

\begin{equation}
\bm{x}_{k+1}=\bm{x}_k-\Delta t\,\bm{v}_{t_k}(\bm{x}_k)+\mathcal{DC}\!\left(\bm{x}_k-t_k\bm{v}_{t_k}(\bm{x}_k), \bm{y}, \bm{A}\right),
\end{equation}

and the corresponding state update with velocity-reuse steps as

\begin{equation}
\tilde{\bm{x}}_{k+1} = \tilde{\bm{x}}_k-\Delta t\,\bm{v}_{t_{k_0}}(\bm{x}_{k_0})+\mathcal{DC}\!\left(\tilde{\bm{x}}_k-t_k\bm{v}_{t_{k_0}}(\bm{x}_{k_0}), \bm{y}, \bm{A}\right), \quad k=k_0,\ldots,k_0+M-1 ,
\end{equation}
where $\mathcal{DC}(\bm{x}, \bm{y}, \bm{A})$ is the data consistency operation given measurement $\bm{y}$ and forward operator $\bm{A}$. 

Here the same velocity $\bm{v}_{t_{k_0}}(\bm{x}_{k_0})$ is reused for $M$ consecutive steps. Assume $\mathcal{DC}(\bm{x}, \bm{y}, \bm{A})$ is $\kappa$-non-expansive with respect to its image argument,

\begin{equation}
\|\mathcal{DC}(\bm{x}, \bm{y}, \bm{A})-\mathcal{DC}(\bm{x}', \bm{y}, \bm{A})\| \le \kappa\|\bm{x}-\bm{x}'\|, \quad 0\le\kappa\le1.
\end{equation}

Let $\bm{\delta}_k=\tilde{\bm{x}}_k-\bm{x}_k$. Then
\begin{equation}
\begin{aligned}
\|\bm{\delta}_{k+1}\| &\le \|\bm{\delta}_k\| + \Delta t \|\bm{v}_{t_k}(\bm{x}_k)-\bm{v}_{t_{k_0}}(\bm{x}_{k_0})\| \\
&\quad+ \kappa \left\| \bm{\delta}_k - t_k \left(\bm{v}_{t_{k_0}}(\bm{x}_{k_0})-\bm{v}_{t_k}(\bm{x}_k)\right)\right\| \\
&\le(1+\kappa)\|\bm{\delta}_k\| + (\Delta t+\kappa t_k) \underbrace{\|\bm{v}_{t_k}(\bm{x}_k)-\bm{v}_{t_{k_0}}(\bm{x}_{k_0})\|}_{\text{velocity mismatch}}.
\end{aligned}
\end{equation}

Using Lipschitz continuity, we have

\begin{equation}
\begin{aligned}
\underbrace{\|\bm{v}_{t_k}(\bm{x}_k)-\bm{v}_{t_{k_0}}(\bm{x}_{k_0})\|}_{\text{velocity mismatch}}
&\le L_x \underbrace{ \|\bm{x}_k-\bm{x}_{k_0}\|}_{\text{state drift}}+L_t|t_k-t_{k_0}|.
\end{aligned}
\end{equation}

Assume the practical updates have bounded increments,

\begin{equation}
\|\bm{x}_{j+1}-\bm{x}_j\|
\le \Delta t\,U_{\max},
\end{equation}

which holds when both the velocity field and the DC correction remain bounded along the trajectory. Then

\begin{equation}
\begin{aligned}
\underbrace{\|\bm{x}_k-\bm{x}_{k_0}\|}_{\text{state drift}}
&\le \sum_{j=k_0}^{k-1}\|\bm{x}_{j+1}-\bm{x}_j\| \\
&\le (k-k_0)\Delta t\,U_{\max}.
\end{aligned}
\end{equation}

Since $|t_k-t_{k_0}|=(k-k_0)\Delta t$, we obtain

\begin{equation}
\underbrace{\|\bm{v}_{t_k}(\bm{x}_k)-\bm{v}_{t_{k_0}}(\bm{x}_{k_0})\|}_{\text{velocity mismatch}} \le (k-k_0)\Delta t\,(L_t+L_xU_{\max}).
\end{equation}

Let $C=L_t+L_xV_{\max}$ and $t_{\max}=\max_k t_k$. Then

\begin{equation}
\|\bm{\delta}_{k+1}\| \le (1+\kappa)\|\bm{\delta}_k\| + (\Delta t+\kappa t_{\max})(k-k_0)\Delta t\,C .
\end{equation}

Unrolling over $M$ consecutive reuse steps with $\bm{\delta}_{k_0}=0$ gives

\begin{equation}
\begin{aligned}
\underbrace{\|\bm{\delta}_{k_0+M}\|}_{\text{accumlated deviation}} &\le (\Delta t+\kappa t_{\max})C\Delta t \sum_{j=0}^{M-1}(1+\kappa)^{M-1-j}j  \\
&\le (\Delta t+\kappa t_{\max})C\Delta t (1+\kappa)^{M-1}\frac{M(M-1)}{2}.
\end{aligned}
\end{equation}

Therefore for bounded $M$,

\begin{equation}
\|\bm{\delta}_{k_0+M}\| = O(M^2\Delta t),
\end{equation}
i.e. the accumulated error remains controlled for bounded consecutive reuse steps.

\end{proof}

\newpage
\section{Appendix: Algorithm}
We summarize the complete EFMCT algorithm in \ref{alg:efmct}. 
\begin{algorithm}[h]
\caption{EFMCT: Flow Matching with Velocity Reuse}
\label{alg:efmct}
\begin{algorithmic}[1]
\Require Forward operator $\bm{A}$, measurements $\bm{y}$, FM model $\bm{v}^{\bm{\theta}}_t(\cdot)$,
number of steps $N$, max reuse steps $M$, relaxation factor $\eta$
\State Initialize $t \gets 1$, $\bm{x}_t \sim \mathcal{N}(\bm{0},\bm{I})$, $\Delta t \gets 1/N$ \Comment{Start from Gaussian}
\State $\bm{v}_{\mathrm{prev}} \gets \varnothing$, $r \gets \|\bm{A}\bm{x}_t - \bm{y}\|^2$, $m \gets 0$ \Comment{Residual $r$ and reuse counter $m$}
\For{$i = 1$ to $N$}
    \Statex \textit{// Attempt velocity reuse if it improves data consistency}
    \If{$\bm{v}_{\mathrm{prev}} \neq \varnothing$ \textbf{and} $m < M$}
        \State $\bm{x}_{t-\Delta t}' \gets \bm{x}_t - \Delta t\,\bm{v}_{\mathrm{prev}}$, \quad $\hat{\bm{x}}_0\ \gets \bm{x}_{t-\Delta t}'-t\bm{v}_{\mathrm{prev}}$, \quad $\tilde{r} \gets \|\bm{A}\bm{x}_{t-\Delta t}' - \bm{y}\|^2$ \Comment{Reuse velocity}
        \If{$\tilde{r} \le \eta\, r$}
            \State $\bm{x}_{t-\Delta t} \gets \bm{x}_{t-\Delta t}' + \bm{v}_{\mathcal{DC}}(\hat{\bm{x}}_0,\bm{y},\bm{A})$, \quad $r \gets \tilde{r}$, \quad $m \gets m+1$ \Comment{DC update}
        \Else
            \State $\bm{v}_{\mathrm{prev}} \gets \varnothing$,\quad $m \gets 0$
            \State \textbf{continue}
        \EndIf
    \Else
    \Statex \textit{// Evaluate FM network at current time step}
        \State $\bm{x}_{t-\Delta t}' \gets \bm{x}_t - \Delta t\,\bm{v}^{\bm{\theta}}_t(\bm{x}_t)$, \quad $\hat{\bm{x}}_0 \gets \bm{x}_{t-\Delta t}'-t\bm{v}^{\bm{\theta}}_t(\bm{x}_t)$
        \State $\bm{x}_{t-\Delta t} \gets \bm{x}_{t-\Delta t}' + \bm{v}_{\mathcal{DC}}(\hat{\bm{x}}_0,\bm{y},\bm{A})$, \quad $r \gets \|\bm{A}\bm{x}_{t-\Delta t} - \bm{y}\|^2$
        \State $\bm{v}_{\mathrm{prev}} \gets \bm{v}^{\bm{\theta}}_t(\bm{x}_t)$,\quad $m \gets 0$
    \EndIf
    \State $t \gets t-\Delta t$, \quad $\bm{x}_t \gets \bm{x}_{t-\Delta t}$
\EndFor
\State \Return $\bm{x}_0$
\end{algorithmic}
\end{algorithm}

\section{Appendix: Extra Results}

\begin{table}[t]
\setlength{\tabcolsep}{4pt}
\centering
\caption{Quantitative comparison of reconstruction perceptual quality across different methods. Mean and standard deviation are reported for LPIPS, where lower values indicate better perceptual similarity. The best average result for each setting is shown in \textbf{bold}, and the second best is \underline{underlined}. NFE$^*$ counts only neural network forward evaluations. $\downarrow$ indicates efficiency improvement in NFE and computation time (on RTX4090) of FMCT/EFMCT compared with the most efficient diffusion-based method.}
\begin{tabular}{c c c c c c}
\toprule[0.1pt]
\multirow{2}{*}{Dataset} 
& \multirow{2}{*}{Method}
& 40 views 
& 20 views 
& \multirow{2}{*}{NFE$^*$} 
& \multirow{2}{*}{Time/s} \\

\cmidrule(lr){3-3} \cmidrule(lr){4-4}
& 
& LPIPS $\downarrow$
& LPIPS $\downarrow$
&  &  \\

\midrule[0.05pt]

\multirow{8}{*}{\rotatebox[origin=c]{90}{\textbf{AAPM}}}
& FBP   
& $0.391 \pm 0.014$
& $0.554 \pm 0.021$
& - & 0.02 \\  

& ADMM-TV  
& $0.272 \pm 0.027$
& $0.380 \pm 0.028$

& - & 9.85 \\  

& DPS   
& $0.128 \pm 0.018$
& $\underline{0.166}\pm 0.025$
& 1000 & 142.16 \\  

& MCG   
& $0.145 \pm 0.017$
& $0.208\pm 0.028$
& 1000 & 143.44 \\

& PGDM  
& $0.159 \pm 0.022$
& $0.187 \pm 0.025$
& 100 & 29.79 \\

& DDS   
& $\underline{0.114}\pm 0.013$
& $0.181 \pm 0.021$
& 50 & 3.83 \\

\cmidrule(lr){2-6}

& \cellcolor{lightgray} FMCT  
& \cellcolor{lightgray} $\textbf{0.101} \pm 0.010$
& \cellcolor{lightgray} $\textbf{0.148} \pm 0.015$

& \cellcolor{lightgray} $\underline{25}(\downarrow50\%)$
& \cellcolor{lightgray} $\underline{1.92}(\downarrow50\%)$ \\ 

& \cellcolor{lightgray} EFMCT 
& \cellcolor{lightgray} $0.129 \pm 0.014$
& \cellcolor{lightgray} $0.194 \pm 0.021$
& \cellcolor{lightgray} $\textbf{7}(\downarrow75\%)$
& \cellcolor{lightgray} $\textbf{0.83}(\downarrow78\%)$ \\  

\midrule[0.1pt]

\multirow{8}{*}{\rotatebox[origin=c]{90}{\textbf{Decathlon}}}
& FBP   
& $0.348 \pm 0.018$
& $0.512\pm 0.026$
& - & 0.02 \\  

& ADMM-TV   
& $0.195 \pm 0.025$
& $0.319 \pm 0.021$
& - & 9.85 \\  

& DPS   
& $0.128 \pm 0.037$
& $0.140 \pm 0.038$
& 1000 & 142.16 \\  

& MCG   
& $0.164 \pm 0.034$
& $0.186 \pm 0.026$
& 1000 & 143.44 \\

& PGDM  
& $0.100 \pm 0.129$
& $0.124 \pm 0.032$
& 100 & 29.79 \\

& DDS   
& $\underline{0.072} \pm 0.015$
& $\textbf{0.110} \pm 0.017$
& 100 & 7.65 \\

\cmidrule(lr){2-6}

& \cellcolor{lightgray} FMCT  
& \cellcolor{lightgray} $0.109 \pm 0.021$
& \cellcolor{lightgray} $0.180 \pm 0.026$
& \cellcolor{lightgray} $\underline{50}(\downarrow50\%)$
& \cellcolor{lightgray} $\underline{4.09}(\downarrow47\%)$ \\ 

& \cellcolor{lightgray} EFMCT 
& \cellcolor{lightgray} $\textbf{0.070} \pm 0.011$
& \cellcolor{lightgray} $\underline{0.120} \pm 0.024$
& \cellcolor{lightgray} $\textbf{11}(\downarrow89\%)$
& \cellcolor{lightgray} $\textbf{1.72}(\downarrow78\%)$ \\  

\bottomrule[0.1pt]
\end{tabular}
\label{tab:lpips_result}
\end{table}

\end{document}